\documentclass{amsart}
\usepackage{amsmath,latexsym, graphicx}
\usepackage[psamsfonts]{amssymb}
\usepackage{times}
\usepackage[mathcal]{euscript}
\usepackage{color}
\usepackage{amsfonts}
\usepackage{amssymb}
\usepackage{amsmath}
\usepackage{amsthm}
\usepackage{bm}
\usepackage[english]{babel}
\usepackage[bookmarks]{hyperref}
\numberwithin{equation}{section} \textwidth=140mm \textheight=200mm
\renewcommand{\epsilon}{\varepsilon}
\renewcommand{\epsilon}{\varepsilon}

\renewcommand{\d}{\mathrm{d}}
\renewcommand{\hat}{\widehat }

\newcommand{\miqdor}{\mathrm{n}}

\newcommand{\black}{\color{black}}

\newcommand{\be}{\begin{equation}}
\newcommand{\ee}{\end{equation}}

\newcommand{\R}{\mathbb{R}}

\newcommand{\T}{\mathbb{T}}

\renewcommand{\H}{\mathbb{H}}
\newcommand{\V}{\mathbb{V}}

\newcommand{\Z}{\mathbb{Z}}
\newcommand{\cC}{{\mathcal C}}

\newcommand{\cE}{{\mathcal E}}
\newcommand{\cF}{{\mathcal F}}
\newcommand{\cG}{{\mathcal G}}

\renewcommand{\Re}{{\ensuremath{\mathrm{Re}}}}

\renewcommand{\det}{\mathop{\mathrm{det}}}

{\bf}{\it}
\newtheorem{theorem}{Theorem}[section]
\newtheorem{lemma}[theorem]{Lemma}

\newtheorem{proposition}[theorem]{Proposition}

\date{\today}

\begin{document}

\large
\title[Number of bound states of the Hamiltonian of a lattice two-boson system...]{Number of bound states of the Hamiltonian of a lattice two-boson system with interactions up to the next neighbouring sites}

\author{Saidakhmat N.~Lakaev, Shakhobiddin I. Khamidov and Mukhayyo O.~Akhmadova}

\address[Saidakhmat Lakaev]{Samarkand State University, 140104, Samarkand, Uzbekistan}
\email{slakaev@mail.ru}

\address[Shahobiddin Khamidov]{Institute of mathematics}
\email{shoh.hamidov1990@mail.ru}

\address[Mukhayyo Akhmadova]{Institute of mathematics}
\email{mukhayyo.akhmadova@mail.ru}

%\maketitle

\begin{abstract}
We study the family $H_{\gamma  \lambda  \mu}(K)$, $K\in\T^2,$ of discrete Schr\"odinger operators,  associated to the Hamiltonian of a system of two identical bosons on the two-dimen\-sional lattice $\mathbb{Z}^2,$ interacting through on one site, nearest-neighbour sites and next-nearest-neighbour sites with interaction magnitudes $\gamma,\lambda$ and $\mu,$ respectively.
We prove there existence an important invariant subspace of operator $H_{\gamma  \lambda  \mu}(0)$ such that the restriction of the operator $H_{\gamma  \lambda  \mu}(0)$ on this subspace has at most  two eigenvalues lying both as below the essential spectrum as well as above it, depending on the interaction magnitude $\lambda,\mu\in \R$ (only). We also  give a sharp lower bound for the number of eigenvalues of $H_{\gamma\lambda\mu}(K)$.
\end{abstract}

\keywords{Two-particle system, lattice Schr\"odinger operator,
essential spectrum, bound states, Fredholm determinant}

\maketitle

\section{Introduction}
Lattice models have an important role in various branches of physics. Among these models are the lattice few-particle Hamiltonians
\cite{Mattis:1986}, which can be viewed as a minimalist version of the
corresponding Bose- or Fermi-Hubbard model involving a fixed finite number of particles of a certain type.  It is well known that the few-particle  lattice Hamiltonians are of a great theoretical interest already in their own
right
\cite{ALzM:2004,ALKh:2012,BPL:2017,KhLA:2021,
LAbdukhakimov:2020,LO'zdemir:2016,LDKh:2016,Lakaev:1993,LLakaev:2017,LA:2021,BKh22,LA23Lob,LIK23}.
Furthermore, these discrete Hamiltonians may be viewed as a natural
approximation for their continuous counterparts
\cite{FMerkuriev:1993} allowing to study few-particles phenomena in the
context of the theory of bounded operators.

One of the intriguing
phenomenon in the spectral theory of three-particle system is the celebrated Efimov effect \cite{Efimov:1970} which
is proven to take place not only in the continuous case but also in the lattice three-body problems \cite{ALzM:2004,
ALKh:2012,DzMSh:2011,Lakaev:1993}. Moreover, the discrete
Schr\"odinger operators represent the simplest and natural model for
description of few-particle systems formed by particles traveling
through periodic structures, say, for ultracold atoms injected into
optical crystals created by the interference of counter-propagating
laser beams \cite{Bloch:2005, Winkler:2006}.

The study of ultracold several-atom systems in optical lattices became very popular in the last
years since these systems possess highly controllable parameters, as lattice  geometry and dimensionality,
two-body potentials,  particle masses, temperature etc.  (see e.g., \cite{Bloch:2005,
J-Z:1998,JZoller:2005,LSAhufinger:2012} and references therein).

Unlike the traditional condensed matter systems, where stable
composite objects are generally formed by attractive forces, the
controllability of the ultracold atomic systems in an optical lattice gives an opportunity to experimentally observe a stable
repulsive bound pair of ultracold atoms, see e.g.,
\cite{Ospelkaus:2006, Winkler:2006}. In all these
observations Bose-Hubbard Hamiltonians became a link between the
theoretical basis and experimental results.

Single-particle Hamiltonians on a one-dimensional
lattice are already of interest in applications. For example,
in \cite {Motovilov:2001}, a one-dimensional single-particle
lattice Hamiltonian was effectively used to demonstrate how
the arrangement of molecules of a certain class in lattice structures
can increase the probability of nuclear fusion.

Contrasting  to the continuous case, the lattice few-particle system does not admit separation of the center-of-mass motion. However, the discrete
translation invariance allows one to use the Floquet-Bloch
decomposition (see, e.g., \cite[Sec. 4]{ALMM:2006}).

Particularly, the total $n$-particle lattice Hamiltonian $\mathrm{H}$ in the (quasi)momentum space representation can be decomposed as the von Neumann direct integral
\begin{equation}
\label{HK} \mathrm{H}\simeq\int\limits_{K\in \T^d} ^\oplus  H(K)\,d
K,
\end{equation}
where $\T^d$ is the $d$-dimensional torus. The fiber Hamiltonian $H(K)$ acting in the respective functional Hilbert space on $\T^{(n-1)d}$ is non  trivially dependent on the quasimomentum $K\in\T^d$ (see e.g.,
\cite{ALMM:2006, FICarroll:2002,Mogilner:1991,Mattis:1986}).

Another reason for studying discrete Hamiltonians is that they  provide a natural approximation for their continuous analogues \cite {FMerkuriev:1993}, which allows to study few-particle systems.

Discrete Schr\"odinger operators also represent a simple model
for description of few-particle systems formed by particles traveling through periodic structures, such as ultracold atoms injected into
optical crystals created by the interference of counter-propagating
laser beams \cite{Bloch:2005, Winkler:2006}. The study of ultracold few-atom systems in optical lattices have became popular
in the last years due to availability of controllable
parameters, such as temperature, particle masses, interaction potentials etc. (see e.g.,
\cite{Bloch:2005,J-Z:1998, JZoller:2005, LSAhufinger:2012}
and references therein). These possibilities allow the opportunity to experimentally observe stable repulsively bound pairs of ultracold atoms (\cite{Ospelkaus:2006,Winkler:2006}, where stable composite objects, as a rule, are formed by using the attractive forces.

Lattice Hamiltonians are of particular interest in fusion physics too. For example, in \cite{Motovilov:2001}, a one-particle one-dimensional lattice Hamiltonian has been successfully employed to show that certain class of molecules in lattice structures may enhance nuclear fusion probability.

The so called fiber Hamiltonians $H(K)$ acting in the  Hilbert space on $\T^{(n-1)d}$ nontrivially depends on the quasimomentum $K\in\T^d$ (see e.g., \cite{ALzM:2004, ALMM:2006}).
The decomposition allows us to reduce the study of the Hamiltonian H to the problem of studying fibered operators $H(K)$.

In this work, we study the spectral properties of the fiber Hamiltonians $H(K),\,K\in\T^2$ acting in the Hilbert space $L^{2,e}(\T^2)$ as
\begin{equation}\label{fiberHamilt}
H_{\gamma\lambda\mu}(K):=H_0(K) + V_{\gamma\lambda\mu},
\end{equation}
where $H_0(K)$ is the fiber kinetic-energy operator,
$$
\bigl(H_0(K)f\bigr)(p)=\cE_K(p)f(p),
$$
with
\begin{equation}\label{def:dispersion}
\cE_K(p):= 2 \sum_{i=1}^2\Big(1-\cos\tfrac{K_i}2\,\cos p_i\Big)
\end{equation}
and $V_{\gamma\lambda\mu}$ is the combined interaction potential. The parameters  $\gamma$, $\lambda$ and $\mu$, called coupling constants, describe interactions between the particles which are located on one site,  on nearest-neighboring-sites and on next-nearest-neighboring-sites of the lattice, respectively.

Notice that $V_{\gamma\lambda\mu}$ does not
depend on $K$ at all. Surely, the operators $H_0(K)$ and
$V_{\gamma\lambda\mu}$, $\gamma, \lambda,\mu\in\R$ are both bounded and
self-adjoint. Since $V_{\gamma\lambda\mu}$ is finite rank, the essential
spectrum of  $H_{\gamma\lambda\mu}(K)$  coincides with the spectrum  of $H_0(K)$, i.e.  $[\cE_{\min}(K),\allowbreak\cE_{\max}(K)],$
where
$$
\cE_{\min}(K):= 2\sum\limits_{i=1}^2\Big(1-\cos \tfrac{K_{i}}2\Big),
\qquad \cE_{\max}(K):= 2\sum\limits_{i=1}^2\Big(1+\cos
\tfrac{K_{i}}2\Big).
$$
To the best of our knowledge, the fiber Hamiltonian \eqref{fiberHamilt} represents
a new exactly solvable model.

The subspaces
\begin{equation}\label{subsp1}
L^{2,e,s}(\T^2):=\Big\{f\in L^{2,e}(\T^2):\,\, f(p_1,p_2)
=f(p_2,p_1),\,\, p_1,p_2\in\T\Big\}
\end{equation}
and
\begin{equation}\label{subsp2}
L^{2,e,a}(\T^2):=\Big\{f\in L^{2,e}(\T^2):\,\, f(p_1,p_2)
=-f(p_2,p_1),\,\,p_1,p_2\in\T\Big\}
\end{equation}
are invariant subspaces of the operator $H_{\gamma\lambda\mu}(0)$ and reduces it.
Denote by $H^s_{\gamma\lambda\mu}(0)$ and $H^a_{\gamma\lambda\mu}(0)$ the restrictions of $H_{\gamma\lambda\mu}(0)$ onto subspaces $L^{2,e,s}(\T^2)$ and $L^{2,e,a}(\T^2)$, respectively.

Note that, the operator $H^a_{\gamma\lambda\mu}(0)$ depends only on the parameters $\lambda$ and $\mu$.
Then
\begin{equation}\label{separation_K0}
\sigma(H_{\gamma\lambda\mu}(0)) = \sigma(H^s_{\gamma\lambda\mu}(0)) \cup
\sigma(H^a_{\lambda\mu}(0)).
\end{equation}

The first purpose of this work is to find the exact number and location of the eigenvalues of the operator $H^{a}_{\lambda\mu}(0)$ lying outside the essential spectrum.
Applying this we find for all $K\in\T^2$  sharp estimates for the number of eigenvalues of the operator $H_{\gamma\lambda\mu}(K)$.

To do these we investigate the Fredholm determinant    $\Delta^a_{\lambda\mu}(z)$ associated to the discrete Schr\"odinger operator $ H^{a}_{\lambda\mu}(0)$. It is well known that there is a one-to-one mapping between the sets of eigenvalues of the operator $H^a_{\lambda\mu}(0)$ and the zeros of $\Delta^a_{\lambda\mu}(z)$ (see \cite{LKhKh:2021}).

Note that the number of zeros of the determinant $\Delta^a_{\lambda\mu}(z)$ lying below (resp.  above) the essential spectrum changes if and only if the coefficient $C^{-}(\lambda, \mu)$ (resp. $C^{+} ( \lambda, \mu)$) of the asymptotics of  $\Delta^a_{\lambda\mu} (z)$ vanishes as $z$ approaches the lower (resp. upper) edge of the essential spectrum (see Lemma \ref{lem:asymp_det}).

Using this property, we establish a partition of the two-dimensional $(\lambda,\mu)$-space into three disjoint connected components by means of curves $C^{-}(\lambda, \mu)=0$ or
$C^{+}(\lambda,\mu)=0$. This allows us to prove that the number of zeros of the determinant $\Delta^a_{\lambda\mu}(z)$ is constant in each connected component.

In \cite{LKhKh:2021,LO'zdemir:2016,BKhL:2022,LKh:2022}, it was  studied the Schr\"odinger operators on the lattice $\Z^d,\,d=1,2$, associated to a system of two bosons with the zero-range on one site interaction ($\lambda\in \R$)  and  interaction on the nearest neighbouring sites ($\mu\in \R$).

For this, we apply the results obtained for the operator $H_{\gamma\lambda\mu}(0)$ and the nontrivial dependence of the dispersion relation $\cE_K$ on the (quasi)momentum $K\in\T^2$.

In \cite{LKhKh:2021} the restriction of the operator $H_{\lambda\mu}(0)$ onto the subspace $L^{2,e,a}(\T^2)$ can have at most one  eigenvalues, lying as below the essential spectrum, as well as above it.

\black
The paper is organized as follows. In Sec. \ref{sec:hamiltonian}
we introduce the two-particle Hamiltonian in the position and
(quasi)momentum representations. Sec. \ref{sec:main_results} contains
statements of our main results. The proofs are presented in Sec. \ref{sec:proofs}.
\black

\section{The two-particle Hamiltonians of a system of two-particles on lattices} \label{sec:hamiltonian}

\subsection{The two-particle Hamiltonian:  the position-space representation}

Let $\mathbb{Z}^2$ be the two dimensional lattice and let
$\ell^{2,s}(\mathbb{Z}^2\times\Z^2)$ be the Hilbert space of square-summable
symmetric functions on $\Z^2\times\Z^2.$

The two particle
Hamiltonian $\hat \H_{\gamma\lambda\mu},$ associated to a system of two
bosons interacting via zero-range and first and second nearest-neighbor potential
$\hat v_{\gamma\lambda\mu}$,   in the position space
representation,
 is a bounded self-adjoint operator
acting in $\ell^{2,s}(\Z^2\times\Z^2)$ as
\begin{equation*}\label{two_total}
\hat \H_{\gamma\lambda\mu}=\hat \H_0 + \hat \V_{\gamma\lambda\mu},\qquad
\gamma,\lambda,\mu\in\R.
\end{equation*}
Here the free Hamiltonian   $\hat \H_0$
 of a system of two identical particles
(bosons) is a bounded self--adjoint operator acting in
$\ell^{2,s}(\Z^2\times\Z^2)$ as
\begin{equation*}
\hat \H_0 \hat f(x,y)= \sum_{n\in\Z^2} \hat \epsilon(x-n) \hat
f(n,y) + \sum\limits_{n\in\Z^2} \hat \epsilon(y-n) \hat f(x,n),
\end{equation*}
where
$$
\hat \epsilon(s) =
\begin{cases}
2 & \text{if $|s|=0,$}\\
-\frac{1}2 & \text{if $|s|=1,$}\\
0 & \text{if $|s|>1,$}
\end{cases}
$$
and $|s|=|s_1|+|s_2|$ for $s=(s_1,s_2)\in \Z^2.$

The interaction $\hat \V_{\gamma\lambda\mu}$ is the multiplication
operator
\begin{equation*}\label{interaction}
\hat \V_{\gamma\lambda\mu} \hat f (x,y) = \hat v_{\gamma\lambda\mu}(x-y) \hat
f(x,y),
\end{equation*}
given by the function
$$
\hat{v}_{\gamma\lambda\mu}(x) =
\begin{cases}
\gamma & |x| = 0,\\
\frac{\lambda}{2} & |x| = 1,\\
\frac{\mu}{2} & |x| = 2,\\
0 & |x|>2.
\end{cases}
$$

\subsection{The two-particle Hamiltonian: the quasimomentum-space representation}

Let $\T^2=( \R /2\pi \Z)^2 \equiv [-\pi,\pi)^2$ be the two
dimensional torus, the Pontryagin dual group of $\Z^2$, equipped
with the Haar measure $\d p.$ Let $L^{2,s}(\T^2\times\T^2)$ be
the Hilbert space of square-integrable symmetric functions on
$\T^2\times\T^2.$ Let $\cF:\ell^2(\Z^2)\rightarrow L^2(\T^2)$ be the
standard Fourier transform
$$
\cF \hat f(p)=\frac{1}{2\pi} \sum_{x\in\Z^2} \hat f(x) e^{ip\cdot
x},
$$
%with the inverse
%$$
%\cF^* \hat f(x):=\int_{\T^2} e^{-i p\cdot x} \hat f(p)\,\d p),
%$$
where $p\cdot x: = p_1x_1+p_2x_2$ for $p=(p_1,p_2)\in\T^2$ and
$x=(x_1,x_2)\in\Z^2.$

In the quasimomentum-space representation the two-particle
Hamiltonian acts in $L^{2,s}(\T^2\times \T^2)$ as
$$
\H_{\gamma\lambda\mu}:=(\cF\otimes\cF) \hat
\H_{\gamma\lambda\mu}(\cF\otimes\cF)^*:=\H_0 + \V_{\gamma\lambda\mu}.
$$
Here the free Hamiltonian  $ \H_0=(\cF \otimes \cF) \hat \H_0
(\cF\otimes \cF)^*$  is
the multiplication operator:
$$
\H_0 f(p,q) = [\epsilon(p) + \epsilon(q)]f(p,q),
$$
where
$$
\epsilon(p) := \sum\limits_{i=1}^2 \big(1-\cos p_i),\quad
p=(p_1,p_2)\in \T^2
$$
is the \emph{dispersion relation} of a single boson.
The interaction  $\V_{\gamma\lambda\mu}=(\cF \otimes
\cF)\hat\V_{\gamma\lambda\mu} (\cF\otimes\cF)^*$ is  (partial) integral
operator
$$
\V_{\gamma\lambda\mu} f(p,q) = \frac{1}{(2\pi)^2}\int_{\T^2}
v_{\gamma\lambda\mu}(p-u) f(u,p+q-u)\d u,
$$
where
$$
v_{\gamma\lambda\mu}(p)=
\gamma+\lambda\sum_{i=1}^2\cos p_i
+ \mu \sum_{i=1}^2\cos 2p_i+2\mu\cos p_1\cos p_2,\quad
p=(p_1,p_2)\in \T^2.
$$

\subsection{The Floquet-Bloch decomposition of $\H_{\gamma\lambda\mu}$ and discrete Schr\"odinger operators}\label{subsec:von_neuman}

Since $\hat H_{\gamma\lambda\mu}$ commutes with the representation of the
discrete group $\Z^2$ by shift operators on the lattice, we can
decompose the space $L^{2,s}(\T^2\times\T^2)$ and $\H_{\gamma\lambda\mu}$
into the von Neumann direct integral as
\begin{equation}\label{hilbertfiber}
L^{2,s}(\T^2\times \T^2)\simeq \int\limits_{K\in \T^2} \oplus
L^{2,e}(\T^2)\,\d K
\end{equation}
and
\begin{equation}\label{fiber}
\H_{\gamma\lambda\mu} \simeq \int\limits_{K\in \T^2} \oplus
H_{\gamma\lambda\mu}(K)\,\d K,
\end{equation}
where $L^{2,e}(\T^2)$ is the Hilbert space of square-integrable even
functions on $\T^2$ (see, e.g., \cite{ALMM:2006}).
The fiber operator $H_{\gamma\lambda\mu}(K),$ $K\in\T^2,$ is a
self-adjoint operator  defined in $L^{2,e}(\T^2)$ as
\begin{equation*}
H_{\gamma\lambda\mu}(K) := H_0(K) + V_{\gamma\lambda\mu},
\end{equation*}
where the unperturbed operator $H_0(K)$ is the multiplication
operator by the function
$$
\cE_K(p):= 2 \sum_{i=1}^2\Big(1-\cos\tfrac{K_i}2\,\cos p_i\Big),
$$
and the perturbation  $V_{\gamma\lambda\mu}$ is defined as
\begin{align}\label{moment_poten}
V_{\gamma\lambda\mu}f(p)
&= \frac{\gamma}{4\pi^2}\int\limits_{\T^2}f(q)\d q + \frac{\lambda}{4\pi^2}\sum\limits_{i=1}^2 \cos p_i \int\limits_{\T^2} \cos q_i f(q) \d q  + \frac{\mu}{4\pi^2} \sum\limits_{i=1}^2\cos 2p_i
\int\limits_{\T^2}\cos 2q_if(q)\d q \\ \nonumber
&+\frac{\mu}{2\pi^2}\cos p_1 \cos p_2
\int\limits_{\T^2}\cos q_1 \cos q_2 f(q)\d q
+ \frac{\mu}{2\pi^2}\sin p_1 \sin p_2
\int\limits_{\T^2}\sin q_1 \sin q_2 f(q)\d q.
\end{align}
The parameter $K\in\T^2$ is usually called the
\emph{two-particle quasi-momentum} and the fiber $H_{\gamma\lambda\mu}(K)$
is called the \emph{Schr\"odinger operator} associated to
the two-particle Hamiltonian $\hat \H_{\gamma\lambda\mu}.$

\subsection{The essential spectrum of lattice Schr\"odinger operators} \label{subsec:ess_spec}

Depending on $\gamma,\lambda, \mu \in \R$ the rank of $V_{\gamma\lambda\mu}$  varies but never exceeds seven. Hence, by the classical Weyl theorem for any $K\in\T^2$ the essential spectrum
$\sigma_{\mathrm{ess}}(H_{\gamma\lambda\mu}(K))$  of $H_{\gamma\lambda\mu}(K)$ coincides with the spectrum of $H_0(K):$
\begin{equation}\label{eq:ess_spec}
\sigma_{\mathrm{ess}}(H_{\gamma\lambda\mu}(K)) = \sigma(H_0(K)) =
[\cE_{\min}(K),\cE_{\max}(K)],
\end{equation}
where
\begin{align*}
\cE_{\min}(K):= & \min_{p\in  \T ^2}\,\cE_K(p) = 2\sum\limits_{i=1}^2\Big(1-\cos \tfrac{K_{i}}2\Big)\geq \cE_{\min}(0)=0,\\
\cE_{\max}(K):= & \max_{p\in  \T ^2}\,\cE_K(p)=
2\sum\limits_{i=1}^2\Big(1+\cos \tfrac{K_{i}}2\Big)\leq
\cE_{\max}(0)=8.
\end{align*}

\subsection{The restriction of  $H_{\gamma\lambda\mu}(0)$ onto the space of antisymmetric functions in $L^{2,e}(\T^2)$ }

Let $L^{2,e,s}(\T^2)$ and $L^{2,e,a}(\T^2)$ subspaces are defined in \eqref{subsp1} and \eqref{subsp2}, respectively.

We recall that
$$
L^{2,e}(\T^2)= L^{2,e,s}(\T^2)\oplus L^{2,e,a}(\T^2).
$$
Since  $\varepsilon$ is even and symmetric,  the spaces $L^{2,e,s}(\T^2)$ and $L^{2,e,a}(\T^2)$  are invariant  subspaces of  the operator  $H_{0}(0)$.

The equality
\begin{align*}
2\cos x \cos y &+ 2\cos z \cos t
=(\cos x + \cos z)(\cos y + \cos t)
 + (\cos x - \cos
z)(\cos y - \cos t)
\end{align*}
gives that the operator $V_{\gamma\lambda\mu}$, in \eqref{moment_poten}, can be written as:
\begin{align}\label{moment_poten_2}
V_{\gamma\lambda\mu}f(p)
&=\frac{\gamma }{4\pi^2}\int_{\T^2} f(q)\,\d q  + \frac{\lambda }{8\pi^2}(\cos p_1 + \cos p_2) \int_{\T^2} (\cos q_1
+ \cos q_2)f(q)\,\d q\\ \nonumber
&+ \frac{\lambda }{8\pi^2}(\cos p_1 - \cos p_2) \int_{\T^2} (\cos q_1
- \cos q_2)f(q)\,\d q\\ \nonumber
&+\frac{\mu }{8\pi^2}(\cos 2p_1 + \cos 2p_2) \int_{\T^2} (\cos 2q_1
+ \cos 2q_2)f(q)\,\d q\\ \nonumber
&+\frac{\mu }{8\pi^2}(\cos 2p_1 - \cos 2p_2) \int_{\T^2} (\cos 2q_1
- \cos 2q_2)f(q)\,\d q\\ \nonumber
& +\frac{\mu }{2\pi^2}\cos p_1\cos p_2 \int_{\T^2} \cos q_1\cos q_2f(q)\,\d q\\ \nonumber
&+\frac{\mu }{2\pi^2}\sin p_1\sin p_2 \int_{\T^2} \sin q_1\sin q_2f(q)\,\d q .
\end{align}
The equality \eqref{moment_poten_2} yields that the subspaces $L^{2,e,s}(\T^2)$ and $L^{2,e,a}(\T^2)$  are invariant  with respect to    $V_{\gamma\lambda\mu}$. Thus, the subspaces $L^{2,e,s}(\T^2)$ and $L^{2,e,a}(\T^2)$  are invariant subspaces of the operator $H_{\gamma\lambda\mu}(0)$ and  reduces it. Therefore
\begin{equation}\label{separation_K0}
\sigma(H_{\gamma\lambda\mu}(0)) = \sigma(H^s_{\gamma\lambda\mu}(0)) \cup
\sigma(H^a_{\gamma\lambda\mu}(0)),
\end{equation}
where $H^s_{\gamma\lambda\mu}(0)$ and $H^a_{\gamma\lambda\mu}(0)$ are the restrictions of $H_{\gamma\lambda\mu}(0)$ onto $L^{2,e,s}(\T^2)$
and $L^{2,e,a}(\T^2).$  Using equality \eqref{moment_poten_2}, for any $f\in L^{2,e,a}(\T^2)$, we arrive that
\begin{align}\label{moment_a}
V_{\gamma \lambda \mu }^{a}f(p) =&\frac{\lambda }{8\pi^2}(\cos p_1 - \cos p_2)
\int_{\T^2} (\cos q_1 - \cos q_2)f(q)\,\d q+ \\ \nonumber
& \frac{\mu }{8\pi^2}(\cos 2p_1 - \cos 2p_2)
\int_{\T^2} (\cos 2q_1 - \cos 2q_2)f(q)\,\d q \in L^{2,e,a}(\T^2)
\end{align}
and hence, $H_{\gamma\lambda\mu}^{a}(0) = H_0(0) + V^{a}_{\gamma\lambda\mu}$.
So, we can conclude that the restriction $H^{a}_{\gamma\lambda\mu}(0)$ of the operator $H_{\gamma  \lambda  \mu}(0)$ onto the Hilbert space $L^{2,e,a}(\T^2)$ is depending only on the parameters $\lambda$ and  $\mu$, which is simpler to investigate. Therefore, here and everywhere below we use in place of $H_{\gamma\lambda\mu}^{a}(0)$ and  $V^{a}_{\gamma\lambda\mu}$ the notations $H_{\lambda\mu}^{a}(0)$ and $V^{a}_{\lambda\mu}$, respectively.

Recall that the main purpose of this article is to study the number and location of the eigenvalues of the operator $H^{a}_{\lambda\mu}(0)$ lying outside the essential spectrum.
Further, applying this we find the number of eigenvalues of the operator $H_{\gamma\lambda\mu}(K),$ for all $K\in\T^2.$

\section{Main results}\label{sec:main_results}

First we study the number of eigenvalues of the operator $H^{a}_{\lambda\mu}(0)$.
To do this, we introduce the following functions \begin{align}\label{polynomial:C+}
C^{+}(\lambda,\mu):=1 -\tfrac{(4-\pi)}{\pi}\lambda - \tfrac{2(32-9\pi)}{\pi}\mu+\tfrac{(32-9\pi)}{4\pi}\lambda \mu
\end{align}
and
\begin{align}\label{polynomial:C-}
C^{-}(\lambda,\mu):=1+\tfrac{(4-\pi)}{\pi}\lambda + \tfrac{2(32-9\pi)}{2\pi}\mu+\tfrac{(32-9\pi)}{4\pi}\lambda \mu.
\end{align}

\black
Then we partition $(\lambda,\mu)$-plane
of interaction magnitudes into connected
components by means of the hyperbolas $C^{+}(\lambda,\mu)=0$ and $C^{-}(\lambda,\mu)=0$.

Denote by
\begin{equation}\label{root_0}
\mu_0=\tfrac{4(4-\pi)}{32-9\pi}.
\end{equation}
By using the number $\mu_0$ we introduce the following functions on $\R$:
\begin{equation}
\label{function_lambda}
\lambda^{\pm}(\mu)=\frac{4}{\mu \mp \mu_0} \pm 8.
\end{equation}

\begin{lemma}\label{lem:function_lambda}
The set of points of $\mathbb{R}^2$ satisfying the equation $C^{\pm}(\lambda,\mu)= 0 $   coincides with the graph of the respective function $\lambda=\lambda^{\pm}(\mu)$.
\end{lemma}

\begin{proof}
For convenience we only prove  for the function $C^{+}(\cdot,\cdot)$. For the function $C^{-}(\cdot,\cdot)$ can be proven similarly.

We observe that $C^{+}(\lambda,\mu)$ can be decomposed as follows:
\begin{equation}\label{decompos_C+}
C^{+}(\lambda,\mu)=\frac{32-9\pi}{4\pi}\big[(\mu-\mu_0)(\lambda-8)-4\big].
\end{equation}
Therefore $C^{+}(\lambda,\mu)=0$ implies the inequality $\mu \neq \mu_0$, i.e., the following system of equation
\begin{align}\label{system1}
\begin{cases}
&C^{+}(\lambda,\mu)= 0 \\
&\mu= \mu_0
\end{cases}
\end{align}
has no solutions.
Otherwise \eqref{system1} gives the equality $C^{+}(\lambda,\mu_0)=0$, but the equality \eqref{decompos_C+} implies that $C^{+}(\lambda,\mu_0)=-\frac{32-9\pi}{\pi}\neq 0.$ Thus, \eqref{system1} has no solutions.
\end{proof}

Since \eqref{system1} has no solutions the straight line $\mu= \mu_0$
(resp. $\mu=-\mu_0$)
separate the graph of the function  $\lambda^{+}$
(resp. $\lambda^{-}$) \eqref{function_lambda}
 into  two (different) continuous curves $\{\tau^{+}_1, \tau^{+}_2\}$,
 (resp. $\{\tau^{-}_1, \tau^{-}_2\}$ ),
 where
\begin{align*}
&\tau^{+}_{1}=\{(\lambda,\mu)\in\R^2:
\lambda=\lambda^+(\mu),\,\,\mu<\mu_0\},\\
&\tau^{+}_{2}=\{(\lambda,\mu)\in\R^2:
\lambda=\lambda^+(\mu),\,\,\mu>\mu_0\}
\end{align*}
and
\begin{align*}
&\tau^{-}_{1}=\{(\lambda,\mu)\in\R^2:
\lambda=\lambda^-(\mu),\,\,\mu>-\mu_0\},\\
&\tau^{-}_{2}=\{(\lambda,\mu)\in\R^2:
\lambda=\lambda^-(\mu),\,\,\mu<-\mu_0\}.
\end{align*}

\begin{center}
\includegraphics[scale=0.28]{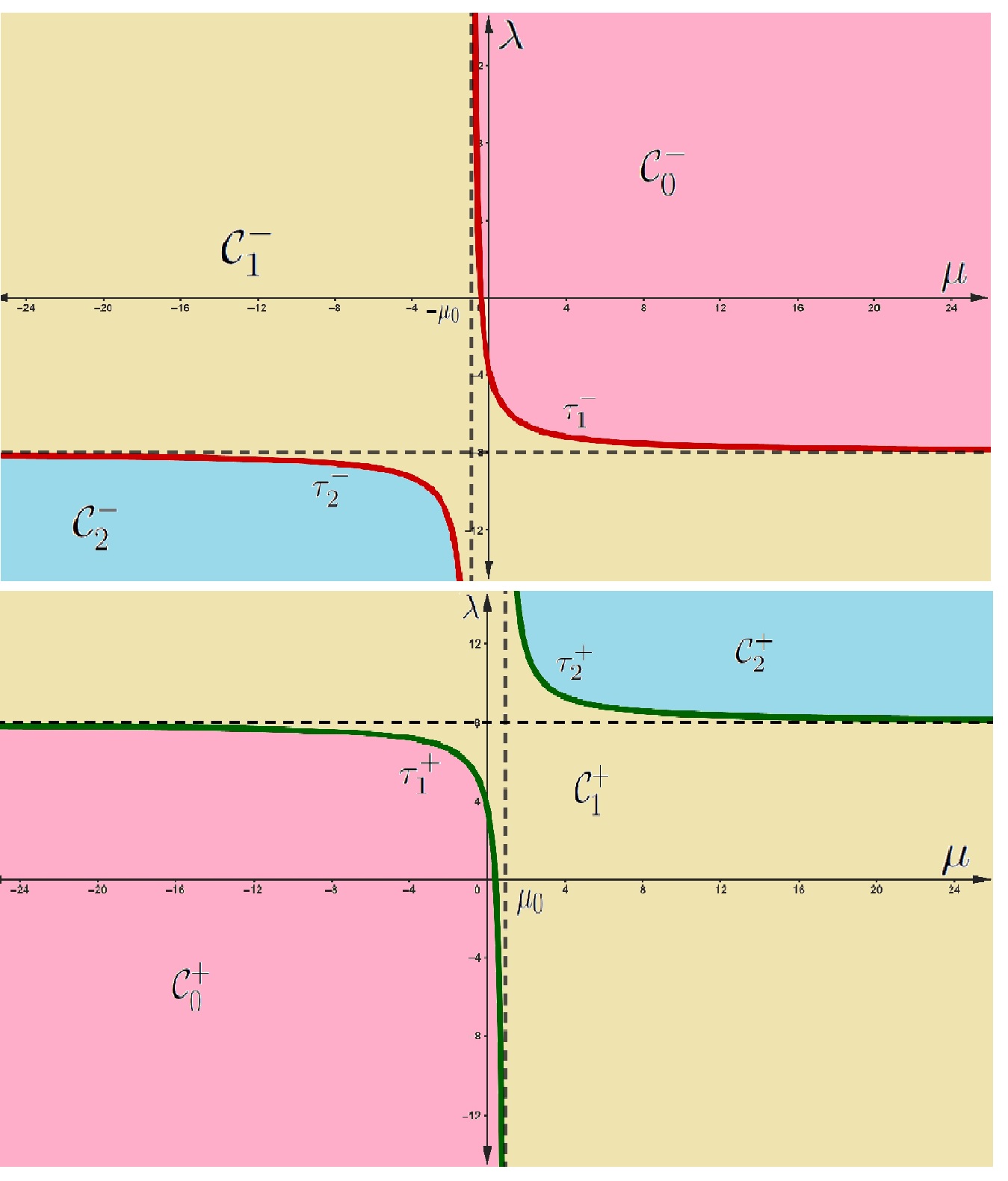}\\
\small{\textbf{Figure\, 3.1.}}
\end{center}

The curves $\{\tau^{\pm}_1,\tau^{\pm}_2\}$ divide the $(\lambda,\mu)$- plane  into the  unbounded connected components $\cC^{\pm}_{0},\cC^{\pm}_{1},\cC^{\pm}_{2}$, where
\begin{align}\label{six_sets}
%antisimmetric right
\cC_0^+:= &\Big\{(\lambda,\mu)\in\R^2:\,\,\lambda<\lambda^{+}(\mu),\,\, \mu<\mu_0 \Big\},\\ \nonumber
\cC_1^+:=&\Big\{(\lambda,\mu)\in\R^2:\,\,\lambda>\lambda^{+}(\mu),\,\, \mu<\mu_0 \Big\}\cup
\Big\{(\lambda,\mu)\in\R^2:\,\,\lambda<\lambda^{+}(\mu),\,\, \mu>\mu_0 \Big\}\\ \nonumber
&\cup
\Big\{(\lambda,\mu)\in\R^2:\,\, \mu=\mu_0 \Big\} ,\\  \nonumber
\cC_2^+:= &\Big\{(\lambda,\mu)\in\R^2:\,\,\lambda>\lambda^{+}(\mu),\,\, \mu<\mu_0 \Big\}
\end{align}
and
\begin{align}\label{six_sets_2}
%antisimmetric right
\cC_0^-:= &\Big\{(\lambda,\mu)\in\R^2:\,\,\lambda>\lambda^{-}(\mu),\,\, \mu>-\mu_0 \Big\},\\ \nonumber
\cC_1^-:=&\Big\{(\lambda,\mu)\in\R^2:\,\,\lambda<\lambda^{-}(\mu),\,\, \mu>-\mu_0 \Big\}\cup
\Big\{(\lambda,\mu)\in\R^2:\,\,\lambda>\lambda^{-}(\mu),\,\, \mu<-\mu_0 \Big\}\\ \nonumber
&\cup
\Big\{(\lambda,\mu)\in\R^2:\,\, \mu=-\mu_0 \Big\} ,\\  \nonumber
\cC_2^-:= &\Big\{(\lambda,\mu)\in\R^2:\,\,\lambda<\lambda^{-}(\mu),\,\, \mu<-\mu_0 \Big\}
\end{align}

The following theorem provides that in each of the above connected components $\cC^{+}_{k}$
(resp. $\cC^{-}_{k}$), $k=0,1,2$ the number of eigenvalues of the operator
$H^{a}_{\lambda\mu}(0)$ lying above (resp. below) its
essential spectrum, remains constant.

\begin{theorem}\label{teo:constant-eea}

Let $\cC^+$ be one of the above connected components
$\cC^+_k$, $k=0,1,2$, of the partition of the
$(\lambda,\mu)$-plane. Then for any $(\lambda,\mu)\in\cC^+$ the
number $n_+({H}^{a}_{\lambda\mu}(0))$ of eigenvalues of
$H^{a}_{\lambda\mu}(0)$ lying above the essential spectrum
$\sigma_\mathrm{ess}\bigl(H^{a}_{\lambda\mu}(0)\bigr)$ remains constant.

Analogously, let $\cC^-$ be one of the above connected
components $\cC^-_k$, $k=0,1,2$, of the partition of the
$(\lambda,\mu)$-plane. Then for any $(\lambda,\mu)\in\cC^-$ the
number $n_-({H}^{a}_{\lambda\mu}(0))$ of eigenvalues of
$H^{a}_{\lambda\mu}(0)$ lying below the essential spectrum
$\sigma_\mathrm{ess}\bigl(H_{\lambda\mu}^{a}(0)\bigr)$ remains constant.
\end{theorem}

The result below concerns the number and location of eigenvalues ({\it with antisymmetric eigenfunctions}) of the
Hamiltonian $H^{a}_{\lambda\mu}(0)$ for various $\lambda,\mu\in\mathbb{R}$.

\begin{theorem}\label{teo:sharpness}
Let $\lambda,\mu\in\mathbb{R}$. The following relations are true:

\begin{itemize}
\item[(a)] if  $(\lambda,\mu)\in \cC_0^+$ resp.  $(\lambda,\mu)\in
\cC_0^-,$  then $H^{a}_{\lambda\mu}(0)$ hasn't any bound states above resp. below the essential spectrum.

\item[(b)] if  $(\lambda,\mu)\in \cC_1^+$ resp.
$(\lambda,\mu)\in \cC_1^-,$  then
$H^{a}_{\lambda\mu}(0)$ has a unique antisymmetric bound state with energy lying above resp.
below the essential spectrum.

\item[(c)] if  $(\lambda,\mu)\in \cC_2^+$ resp.   $(\lambda,\mu)\in \cC_2^-,$
 then $H^{a}_{\lambda\mu}(0)$
has  exactly two antisymmetric bound states with energy lying above resp. below the essential spectrum;
\end{itemize}
\end{theorem}

%\begin{figure}
%\includegraphics[width=0.75\textwidth]{figure7.jpg}
%\caption{ Partition of the $(\lambda,\mu)$-plane of parameters $\lambda,\mu \in \R$ in the connected components $\cG_{\alpha\beta}, \alpha,\beta=0,1,2$.  These components are tagged by the symbols $[N_-,N_+]$ formed of the numbers
%$N_{-}$ and $N_{+}$ of eigenvalues of the operator $H^a_{\lambda\mu}(0)$ lying below and above the essential spectrum, respectively. Until the point $(\lambda,\mu)$ does not cross any of the borders between $\cG_{\alpha\beta}$,  no change occurs in $N_-$ and $N_+$. However, as soon as $(\lambda,\mu)$ crosses one of those borders, the essential spectrum of $H^a_{\lambda\mu}(0)$ either  {\it gives birth} or   {\it absorbs} eigenvalues of $H^{a}_{\lambda\mu}(0)$. }\label{fig:sohalar}
%\end{figure}

Let $\miqdor_+(H_{\gamma\lambda\mu}(K))$ resp.
$\miqdor_-(H_{\gamma\lambda\mu}(K))$ be the number of eigenvalues of the operator  $H_{\gamma\lambda\mu}(K)$
above resp. below the essential spectrum.
Denote by
$$\cG_{\alpha\beta}:=\cC^{-}_{\alpha} \cap \cC^{+}_{\beta},\,\,\,\alpha,\beta=0,1,2$$
 the connected components.

Next theorem  gives  estimates for  the number of eigenvalues of  the operator  $H_{\gamma\lambda\mu}(K)$, for all  $ K\in\T^2$,  depending only to parameters  $\lambda,\mu \in \R.$

\begin{theorem}\label{teo:eigenK}
Let $(\lambda,\mu)\in \cG_{\alpha\beta}$ then
for all $\alpha,\beta=0,1,2$
\begin{align*}
\qquad  n_-(H_{\gamma\lambda\mu}(K)) \ge \alpha  \quad \text{and} \quad n_+(H_{\gamma\lambda\mu}(K)) \ge \beta.
\end{align*}
\end{theorem}

\black

\section{Proof of the main results}\label{sec:proofs}

Let $\{\alpha_1,\alpha_2\} $ be a  system of vectors in
$L^{2,e,a}(\T^2)$, with
\begin{equation}\label{ons2}
\begin{array}{cl}
&\alpha_{1}(p)=\dfrac{\cos p_1-\cos p_2}{\sqrt{2\pi}},\quad  \quad  \quad  \alpha_{2}(p)=\dfrac{\cos  2p_1-\cos 2p_2}{\sqrt{2\pi}}.
\end{array}
\end{equation}

One easily verifies by inspection that the vectors \eqref{ons2}  are orthonormal in $L^{2,e,a}(\T^2)$. By using the
orthonormal systems \eqref{ons2}  one obtains
\begin{align}\label{repr1}
&{V}^{a}_{\lambda\mu}{f}=\frac{\lambda}{4}({f},\alpha_{1})\alpha_{1}+\frac{\mu}{4}({f},\alpha_{2})\alpha_{2}
\end{align}
where  $(\cdot,\cdot)$ is the inner product in $L^{2,e,a}(\T^2).$
For any $z\in\R \setminus[0,\,8]$ we define (the transpose of) the
Lippmann-Schwinger operator (see., e.g., \cite{LSchwinger:1950}) as
\begin{equation*}\label{LSchw1}
{B}_{\lambda\mu}^{a}(0,z)=-{V}_{\lambda\mu}^{a}{R}_0(0,z),
\end{equation*}
where  ${R}_0(0,z):= [{H}_0(0)-zI]^{-1},\,
z\in\mathbb{R}\setminus[0,\,8]$, is the resolvent of the operator
${H}_0(0)$.
\begin{lemma}\label{eigen-eigenvalue}
For each $\lambda,\mu \in \R$  the number $z\in \mathbb{R}\setminus
[0,\,8]$ is an eigenvalue of the operator
${H}_{\lambda\mu}^{a}(0)$ if and only if the number $1$ is an
eigenvalue for
 ${B}_{\lambda\mu}^{a}(0,z)$.
\end{lemma}
The proof of this lemma is quite standard (see., e.g.,
\cite{Albeverio:1988}) and, thus, we omit~it.

The representation \eqref{repr1} yields the equivalence of the
Lippmann-Schwinger equation
\begin{align*}
{B}_{\lambda\mu}^{a}(0,z){\varphi}={\varphi}, \,\,\,{\varphi}
\in L^{2,e,a}(\T^2)
\end{align*}
to the following algebraic linear system in
$x_i:=({\varphi},\alpha_{i}),\,\,i=1,2$:

\begin{equation}\label{system}
\left\lbrace\begin{array}{ccc}
[1+\lambda a_{11}(z)]x_{1}+\mu a_{12}(z)x_2=0, \\
 \lambda  a_{12}(z) x_{1}+ [1+\mu a_{22}(z)]x_{2}=0,
\end{array}\right.
\end{equation}
where
\begin{align}\label{def:functions}
a_{ij}(z):=&  \frac{1}{4}\int_{\T^2} \frac{\alpha_i(p)\alpha_j(p)\d p}{\cE_0(p)-z},\,\, i,j=1,2.
\end{align}

Let $z\in\mathbb{R}\setminus[0,\,8]$ and
\begin{equation}\label{determinant-ea}\Delta^{a}_{\lambda\mu }(z):=\det[I-{B}_{\lambda\mu}^{a}(0,z)]=
\begin{vmatrix}
 1+\lambda  a_{11}(z)&\mu  a_{12}(z)\\
 \lambda  a_{12}(z)& 1+\mu  a_{22}(z)
\end{vmatrix}.
\end{equation}

\black

The following lemma
describes the relations between the operator  $H_{\lambda \mu }^a$ and the
determinant  $\Delta^{a}_{\lambda \mu}(z)$ defined in  \eqref{determinant-ea}.

\begin{lemma}\label{lem:det_zeros_vs_eigen}
Given $\lambda,\mu \in \R$, a number $z\in
\mathbb{R}\setminus[0,8]$ is an eigenvalue of $H_{ \lambda \mu
}^\mathrm{a}(0)$  if and only if it is a
zero of $\Delta^\mathrm{a}_{ \lambda \mu }(\cdot)$. Moreover, in $\R\setminus [0,8]$ the function $\Delta^\mathrm{a}_{ \lambda \mu  }(\cdot)$ has  at most two zeros.
\end{lemma}

The proof of this statement is rather standard (cf e.g.
\cite{LKhKh:2021,LBozorov:2009}) and we skip it.

Main properties of the functions \eqref{def:functions} and
asymptotic behavior  as $z\nearrow 0$ or $z\searrow 8$ is described
in the next asserion.

\begin{proposition}\label{prop:asymp_functions}
The functions  $a_{ij}(z)$ $(i,j=1,2)$ are real-valued for
$z\in\mathbb{R}\setminus[0,8]$, strictly increasing and positive in
$(-\infty,0)$, strictly decreasing and negative in $(8,+\infty)$.
Furthermore, the following asymptotics are valid:
\begin{align*}
&\lim\limits_{z\nearrow0}a_{11}(z)=\frac{4-\pi}{\pi}, \,\,\,\, \quad \quad  \lim\limits_{z\searrow 8}a_{11}(z)=-\frac{4-\pi}{\pi}, \\
&\lim\limits_{z\nearrow0}a_{12}(z)=\frac{32-9\pi}{4\pi}, \quad \quad \lim\limits_{z\searrow 8}a_{12}(z)=-\frac{32-9\pi}{4\pi},\\
&\lim\limits_{z\nearrow0}a_{22}(z)=\frac{32-9\pi}{2\pi}, \quad \quad \lim\limits_{z\searrow 8}a_{22}(z)=-\frac{32-9\pi}{2\pi}.
\end{align*}
\end{proposition}

\begin{proof}
The asymptotic relations for the function
$a_{11}(\cdot)$ have
been established in
 \cite{LKhKh:2021}.  The remaining asymptotic
relations are proven by using  the
following limit equalities:

\begin{align*}
\lim\limits_{z\nearrow 0} \Big( a_{12}(z)-2a_{11}(z)\Big)=&\frac{1}{8\pi^2}\int_{\T^2}\frac{(\cos p_1 - \cos p_2)^2(\cos p_1 + \cos p_2-2) \,\d
p}{\cE_0(p)}=\\
&-\frac{1}{16\pi^2}\int_{\T^2}(\cos p_1 - \cos p_2)^2 \,\d
p=-\frac{1}{4} , 	\\
\lim\limits_{z\nearrow 0} \Big(a_{22}(z)-2a_{12}(z)\Big)=&-\frac{1}{16\pi^2}\int_{\T^2}(\cos 2p_1 - \cos 2p_2)(\cos p_1 - \cos p_2) \,\d
p=0 .
\end{align*}

The  asymptotical formulae for $z > 8$  are derived
analogously and, thus, we skip the respective computation.
\end{proof}

\begin{lemma}\label{lem:asymp_det}
The function $\Delta^{a}_{\lambda \mu }(z)$ is real-valued and well defined on $\mathbb{R}\setminus[0,8]$. Moreover it has following asymptotics:
\begin{itemize}
\item[(a)]$\lim\limits_{z\rightarrow {\pm\infty}}\Delta^a_{\lambda \mu }(z)=1$,
\item[(b)]$\Delta^a_{\lambda \mu }(z) =C^{-}(\lambda,\mu)+o(1)$, as $z\nearrow 0,$
\item[(c)]$\Delta^a_{\lambda \mu }(z) =C^{+}(\lambda,\mu)+o(1)$, as $z\searrow 8,$
\end{itemize}
where $C^{+}(\lambda,\mu)$ resp.$C^{-}(\lambda,\mu)$ defined as \eqref{polynomial:C+} resp. \eqref{polynomial:C-}.
 \end{lemma}
 \black

\begin{proof}

The first item follows from the Lebesgue dominated convergence theorem. The proof of the last items  can be derived by using the asymptotics of the functions $a_{ij}(\cdot),i,j=1,2,$ in Proposition \ref{prop:asymp_functions}.
\end{proof}

\textit{Proof of Theorem \ref{teo:constant-eea} }.
Let  $\cC\in\{\cC^{+}_{0},\cC^{+}_{1},\cC^{+}_{2}\}$.
Since for any $(\lambda,\mu)\in\cC$ the determinant $\Delta^{a}_{\lambda\mu}(z)$ is real analytic in $z \in \{\Re \,z>8\}$
there exist positive numbers $B_2>B_1>8$ such that the function $\Delta^{a}_{\lambda\mu}(z)$ has only finite number zeros in the interval$(B_1,B_2)$.

Let $(\lambda_0,\mu_0)\in \cC$ be a point and $z_0>8$ be a zero of multiplicity $m=0,1,2,...$ for the function $\Delta^{a}_{\lambda_0\mu_0}(z)$.  Since for each fixed $z>8$ the determinant $\Delta^{a}_{\lambda\mu}(z)$ is a real analytic function in $(\lambda,\mu)\in\cC$ and
the function $\Delta^{a}_{\lambda\mu}(z)$ is real-analytic in $z\in(8,+\infty)$, for each $\varepsilon>0$ there are numbers $\delta>0$, $\eta> 0$ and a neighborhood $W_{\eta}(z_0)$ of $z_0>8$ such that for all $z\in \overline{W_{\eta}(z_0)}$ and $( \lambda,\mu )\in\cC$ satisfying the conditions $|z-z_0|=\eta$ and $||(\lambda,\mu)-(\lambda_0,\mu_0)||<\delta$
the following two inequalities $|\Delta^{a}_{\lambda_0\mu_0}(z)|>\eta$ and $|\Delta^{a}_{\lambda\mu}(z)-\Delta^{a}_{\lambda_0\mu_0}(z)|<\epsilon $ are correct.
Then by Rouche's theorem for all $(\lambda,\mu)\in\cC$ satisfying the condition $||(\lambda,\mu)-(\lambda_0,\mu_0)||<\delta $ the number of zeros of the function $\Delta^{a}_{\lambda\mu}(z)$ remains constant in ${U_{\eta}(z_0)}$. Since the zero $z_0>8$ of the function $\Delta^{a}_{\lambda\mu}(z)$ is arbitrary in $(B_1,B_2)$ we conclude that for all $(\lambda,\mu)\in\cC$ satisfying the condition $||(\lambda,\mu)-(\lambda_0,\mu_0)||<\delta$ the number of zeros of the function $\Delta^{a}_{\lambda\mu}(z)$ remains constant.

Moreover each Jordan curve $\Gamma\subset\cC$ connecting any two points of $\cC$ is a {\it compact set}, therefore for any $(\lambda,\mu)\in \Gamma$ the number of zeros of the function $\Delta^{a}_{\lambda\mu}(z)$ lying above eight  remains a constant. Then, Lemma
\ref{lem:det_zeros_vs_eigen} gives that the number of eigenvalues of the
operator ${H}^{a}_{\lambda\mu}(0)$ is a constant.
$\square$

\black
Further, we prove our main results.

\textit{Proof of Theorem \ref{teo:sharpness}}
Since to simmetricity we prove the lemma only for the case $(\lambda,\mu)\in \cC^{+}_{i},i=0,1,2$. The other cases can be proved similarly.

 We note that, by Lemma \ref{lem:asymp_det}
\begin{equation}\label{behav_sim_infty}
\lim\limits_{z\to+\infty} \Delta^{a}_{\lambda\mu}(z)  =1
\end{equation}
and
\begin{align*}
\Delta^{a}_{\lambda\mu}(z)=& C^{+}(\lambda,\mu):= 1 -\tfrac{(4-\pi)}{\pi}\lambda - \tfrac{2(32-9\pi)}{\pi}\mu+\tfrac{(32-9\pi)}{4\pi}\lambda \mu,\,\, as\,\, z\searrow 8.
\end{align*}

(a) According to definitions  \eqref{polynomial:C+}   and \eqref{six_sets} we have that  $(0,0)\in \cC^{+}_{0}$. The representation
\eqref{determinant-ea} of determinant $\Delta^{a}_{\lambda\mu}(z)$ yields that
$$\Delta^{a}_{00}(z)\equiv 1.$$ Therefore the function $\Delta^{a}_{00}(z)$ has no zeros in $(8,+\infty)$.
Lemma \ref{lem:det_zeros_vs_eigen} yields that the operator $H_{00}^{a}(0)$ has no one eigenvalues above the essential spectrum. Theorem \ref{teo:constant-eea} gives that the operator $H_{\lambda\mu}^{a}(0)$ has not any bound states with energy lying above the essential spectrum for all $(\lambda,\mu)\in \cC^{+}_{0}$.

(b) Let $(\lambda,\mu)\in \cC^{+}_{1}$ be an arbitrary point. By using  the definition of $\cC^{+}_{1}$ in \eqref{six_sets}  we get
\begin{align}\label{ineq:b1}
1 -\tfrac{(4-\pi)}{\pi}\lambda - \tfrac{2(32-9\pi)}{\pi}\mu+\tfrac{(32-9\pi)}{4\pi}\lambda \mu<0.
\end{align}

According to relation \eqref{ineq:b1} we have that
\begin{equation}\label{ineq:b2}
\lim\limits_{z\searrow 8}\Delta^{a}_{\lambda\mu}(z)<0 \,\, \text{and}  \, \lim\limits_{z\rightarrow +\infty}\Delta^{a}_{\lambda\mu}(z)=1.
\end{equation}

The expressions \eqref{ineq:b2} show that the continuous function
$\Delta^{a}_{\lambda\mu}(\cdot)$ changes sign in $(8,+\infty),$ so
that the function $\Delta^{a}_{\lambda\mu}(z)=0$ has at least one
zero in $(8,+\infty).$ We show that the function $\Delta^{a}_{\lambda\mu}(z)=0$ has exactly one
zero in $(8,+\infty).$ To the contrary, assume that this function has at least  two zeros.
On the other hand, since $\Delta^{a}_{\lambda\mu}(\cdot)$ has different signs
at the endpoints of $(8,+\infty),$
the function $\Delta^{a}_{\lambda\mu}(\cdot)$ should have at least three
zeros.
This contradicts to Lemma
\ref{lem:det_zeros_vs_eigen}. Hence, the function  $\Delta^{a}_{\lambda\mu}(\cdot)$ must have exactly one zero in $(8,+\infty).$ By Lemma \ref{lem:det_zeros_vs_eigen} the operator $H^{a}_{\lambda\mu}(0)$ has a unique antisimmetric bound state with energy lying above the essential spectrum.

(c) Now we suppose that $(\lambda,\mu)\in \cC^{+}_{2}$ be an arbitrary point. According to \eqref{six_sets}  the following inequalities
\begin{align}\label{ineq1}
1 -\tfrac{(4-\pi)}{\pi}\lambda - \tfrac{2(32-9\pi)}{\pi}\mu+\tfrac{(32-9\pi)}{4\pi}\lambda \mu>0 \quad  \text{and} \quad \mu>\mu_0
\end{align}
 are hold. \black

 \eqref{ineq1} leads that
$$\lambda>\frac{4}{\mu - \mu_0} + 8>8.$$

Proposition \ref{prop:asymp_functions} and inequality $\lambda>8$ obeys that
\begin{align*}
\lim\limits_{z\searrow 8}\Delta^{a}_{\lambda0}(z)=1 -\tfrac{(4-\pi)}{\pi}\lambda<0
 \quad \text{and} \quad \lim\limits_{z\rightarrow +\infty}\Delta^{a}_{\lambda0}(z)=1.
\end{align*}

Therefore the continuous and monotone increasing function $\Delta^{a}_{\lambda0}(z)=1+\lambda a_{11}(z)$
has a unique zero $z_{11}$ above the essential spectrum. Then we observe that
\begin{align}\label{ineq2}
\Delta^{a}_{\lambda\mu}(z_{11})=(1+\lambda  a_{11}(z_{11}))(1+\mu  a_{22}(z_{11}))-\lambda\mu  (a_{12}(z_{11}))^2=-\lambda\mu (a_{12}(z_{11}))^2<0.
\end{align}

The inequalities \eqref{ineq1} and \eqref{ineq2} yield the following relations
$$
\lim\limits_{z\searrow 8}\Delta^{a}_{\lambda\mu}(z)>0, \quad  \Delta^{a}_{\lambda\mu}(z_{11})<0  \quad \text{and} \quad  \lim\limits_{z\rightarrow +\infty}\Delta^{a}_{\lambda\mu}(z)=1,
$$
which implies that the continuous function $\Delta^{a}_{\lambda\mu}(z)$ has at least one zero in each intervals $(8,z_{11})$ and
$(z_{11},+\infty)$. Therefore there exists  real numbers satisfying the inequalities  $8<z_{22}<z_{11}<z_{21}$
such that
\begin{align*}
\Delta^{a}_{\lambda\mu}(z_{21})=\Delta^{a}_{\lambda\mu}(z_{22})=0.
\end{align*}
Lemma \ref{lem:det_zeros_vs_eigen} implies that  the operator $H^{a}_{\lambda\mu}(0)$ has exactly two antisymmetric bound states with energy lying  above the essential spectrum.
 \hfill $\square$

\textit{ Theorem \ref{teo:eigenK}} can be proved as  Theorem 3.2 in \cite{LKhKh:2021}. $\square$

\black

\end{document}